\documentclass[journal,letterpaper]{IEEEtran}

\usepackage[utf8]{inputenc} 
\usepackage[T1]{fontenc}
\usepackage{url}
\usepackage{ifthen}
\usepackage{cite}
\usepackage[cmex10]{amsmath} 

\usepackage{amsfonts}
\usepackage{algorithmic}
\usepackage{algorithm}
\usepackage{array}
\usepackage[caption=false,font=normalsize,labelfont=sf,textfont=sf]{subfig}
\usepackage{textcomp}
\usepackage{stfloats}
\usepackage{verbatim}
\usepackage{graphicx}
\usepackage{amsmath}
\usepackage{amssymb}
\usepackage{amsthm}
\usepackage{graphicx}
\usepackage{epstopdf}
\usepackage{multirow}
\usepackage{stmaryrd}
\usepackage[colorlinks=true,linkcolor=blue,citecolor=blue]{hyperref}
\newtheorem{theorem}{Theorem}
\newtheorem{lemma}{Lemma}
\newtheorem{proposition}{Proposition}

\theoremstyle{definition}
\newtheorem{definition}{Definition}
\newtheorem{remark}{Remark}
\newtheorem{example}{Example}
\makeatletter
\def\ps@IEEEtitlepagestyle{%
  \def\@oddhead{\relax}%
  \def\@evenhead{\relax}%
  \def\@oddfoot{\relax}%
  \def\@evenfoot{\relax}}
\makeatother
\begin{document}
\thispagestyle{empty}
\title{A Parity-Consistent Decomposition Method for the Weight Distribution of Pre-Transformed Polar Codes} 


\author{$\text{Yang Liu}^{\ast}$, $\text{Bolin Wu}^{\dagger}$, $\text{Yuxin Han}^{\ast}$ and $\text{Kai Niu}^{\ast}$ \\
        $^{\ast}\text{Key}$ Laboratory of Universal Wireless Communication, Ministry of Education \\
        Beijing University of Posts and Telecommunications, Beijing, China \\
	  $^{\dagger}\text{School}$ of Mathematical Sciences, Beijing University of Posts and Telecommunications, Beijing, China \\
        \{liuyang, bolinwu, hanyx, niukai\}@bupt.edu.cn
\thanks{This work was supported by the National Natural Science Foundation of China under Grant 62471054. \textit{(Corresponding author: Kai Niu.)}}}

\maketitle


\begin{abstract}
This paper introduces an efficient algorithm based on the Parity-Consistent Decomposition (PCD) method to determine the WD of pre-transformed polar codes. First, to address the bit dependencies introduced by the pre-transformation matrix, we propose an iterative algorithm to construct an \emph{Expanded Information Set}. By expanding the information bits within this set into 0s and 1s, we eliminate the correlations among information bits, thereby enabling the recursive calculation of the Hamming weight distribution using the \emph{PCD method}. Second, to further reduce computational complexity, we establish the theory of equivalence classes for pre-transformed polar codes. Codes within the same equivalence class share an identical weight distribution but correspond to different \emph{Expanded Information Set} sizes. By selecting the pre-transformation matrix that minimizes the \emph{Expanded Information Set} size within an equivalence class, we optimize the computation process. Numerical results demonstrate that the proposed method significantly reduces computational complexity compared to existing deterministic algorithms.
\end{abstract}


\section{Introduction}
\IEEEPARstart{P}{olar} codes are the first class of capacity-achieving codes \cite{5075875} and have been adopted as the coding scheme for control channels in the 5G NR standard \cite{3gpp38212}. However, in the short-to-medium block length regime, their error-correction performance still lags behind the finite-length capacity limit\cite{Polyanskiy2010Channel}. To bridge this gap and enhance the error-correction capability at finite lengths, various pre-transformed polar coding schemes have been proposed, including CRC-concatenated polar codes\cite{Niu2012CRC}, parity-check polar codes\cite{Zhang2018Parity}, polar subcodes\cite{Trifonov2016Polar}, and PAC codes\cite{Arikan2019From}.

To better predict the error-correction performance of pre-transformed polar codes, recent research has focused on analyzing their Hamming weight distributions. Typically, the first few terms of the weight spectrum suffice to accurately estimate the ML performance in the high SNR regime. Specifically, \cite{Piao2020Sphere} utilized a sphere decoding algorithm to enumerate the minimum Hamming weight codewords of CRC-concatenated polar codes, while the minimum weight codewords of PAC codes have been investigated in \cite{Rowshan2022Fast,rowshan2023formation,Ellouze2023Low}. For general pre-transformed polar codes, \cite{Zunker2024Enumeration} proposed a tree intersection method to enumerate the minimum weight codewords, and \cite{Miloslavskaya2022Computing} calculated the number of low-weight codewords via an $X4$-based construction. However, computing the complete Hamming weight distribution of pre-transformed polar codes remains a more challenging task. \cite{Li2021Weight} derived the average Hamming weight distribution for random pre-transformed polar codes, but this yields only an approximate result. More recently, \cite{yao2024deterministic} achieved the exact recursive calculation of the complete weight distribution by representing the polar code as a union of polar cosets.

However, existing algorithms for computing the complete Hamming weight distribution of pre-transformed polar codes still suffer from high computational complexity. To mitigate this, we further exploit the structural properties of pre-transformed polar codes based on the \emph{PCD method}\cite{Liu2506:Parity}. Fig. \ref{fig1} illustrates the block diagram of the proposed computation framework. First, for any arbitrary pre-transformed polar code, we identify a corresponding equivalence class characterized by an identical Hamming weight distribution, where each code within the class possesses a distinct pre-transformation matrix. Subsequently, we calculate the size of the expanded information set, denoted by $\lambda$, for each code in the equivalence class. It is worth noting that $\lambda$ is a critical parameter that exponentially influences the computational complexity. Finally, the pre-transformed polar code corresponding to the minimum $\lambda$ is selected to compute the complete Hamming weight distribution via the \emph{PCD method}.

\textit{Notation Conventions:} Throughout this paper, we use $u_{1}^{N}$ to denote an $N$-dimensional vector and $u_{i}$ to denote the $i$-th element in $u_{1}^{N}$. We write $u_{\mathcal{A}}$ to denote a subvector of $u_{1}^{N}$ given by elements with indices from the set $\mathcal{A}$. The notation ${0}_{1}^{N}$ is used to denote the all-zero vector. We use $u_{1}^{N} \circledast v_{1}^{N}$ to denote the convolution of $u_{1}^{N}$ and $v_{1}^{N}$. We write $\llbracket a,b\rrbracket$ to denote the set $\left\{i|a\leq i\leq b,i\in \mathbb{Z}\right\}$.
\section{Preliminaries of Polar Codes}
\subsection{Pre-Transformed Polar Codes}
Given $N=2^{n}$, an $(N,K)$ polar code is generated by the $N$-dimensional matrix $G_{N}=K_{2}^{\otimes n}$, where ${K}_{2}^{\otimes n}$ is the $n$-th Kronecker power of kernel ${K}_{2}=\begin{bmatrix}\begin{smallmatrix}1 & 0 \\1 & 1 \end{smallmatrix}\end{bmatrix}$. Any codeword is expressed as $x_{1}^{N}=u_{1}^{N}G_{N}$, where $u_{1}^{N}$ is the uncoded sequence comprising $K$ information bits and $(N-K)$ frozen bits. The information set $\mathcal{I}$ and frozen set $\mathcal{F}$ satisfy the relation $\mathcal{I}\cup \mathcal{F}=\llbracket 1,N\rrbracket$. \\ 

Let $T$ be an $N\times N$ upper-triangular pre-transformation matrix. Then the codeword of pre-transformed polar code is given by $x_{1}^{N}=u_{1}^{N}TK_{2}^{\otimes n}$. We denote the vector $u_{1}^{N}T$ by $v_{1}^{N}$.
\subsection{PCD Method for Polar Codes}
In the \emph{PCD method}\cite{Liu2506:Parity}, a polar code $\mathcal{C}$ is represented as a union of \emph{PCD coset} with the following structure:
\begin{equation}\label{2-1}
  \mathcal{C}(\hat{u}_{\mathcal{P}})\triangleq\{u_{1}^{N}G_{N}|u_{\mathcal{P}}=\hat{u}_{\mathcal{P}},u_{\mathcal{I}\backslash \mathcal{P}}\in \{0,1\}^{|\mathcal{I}|-|\mathcal{P}|}\},
\end{equation}
where set $\mathcal{P} \triangleq \{i|i\in \mathcal{I}, i \text{ mod } 2 =0, i-1 \in \mathcal{F}\}$. When only one of $u_{2m-1}$ and $u_{2m}$ is information bit, the index of the information bit is record in $\mathcal{P}$. It would not happen that $u_{2m-1}$ is an information bit and $u_{2m}$ is a frozen bit since polar codes are decreasing monomial codes.  

The coset in \eqref{2-1} satisfies that $u_{2m-1}$ and $u_{2m}$ are both information bits or frozen bits. The codeword of the coset can be represented as $c=(c_{1},c_{2})$, where $c_{1}$ and $c_{2}$ are codewords of polar codes $\mathcal{C}_{1}(\hat{u}_{\mathcal{P}})$ and $\mathcal{C}_{2}(\hat{u}_{\mathcal{P}})$, respectively. $\mathcal{C}_{1}(\hat{u}_{\mathcal{P}})$ and $\mathcal{C}_{2}(\hat{u}_{\mathcal{P}})$ is abbreviated to $\mathcal{C}_{1}$ and $\mathcal{C}_{2}$ in unambiguous contexts. The generator matrix of $\mathcal{C}_{1}$ and $\mathcal{C}_{2}$ is $K_{2}^{\otimes n-1}$. The information set and frozen set of $\mathcal{C}_{1}$ and $\mathcal{C}_{2}$ are given by
\begin{align}
  \mathcal{I}_{s}&\triangleq \{m|2m-1\in \mathcal{I}, 2m\in \mathcal{I}\}, \label{2-2}\\
  \mathcal{F}_{s}&\triangleq \{m|2m-1\in \mathcal{F}, 2m\in \mathcal{F}\}.\label{2-3}
\end{align}
The value of frozen bits in the uncoded sequence $u^{(1)}$ and $u^{(2)}$ corresponding to $\mathcal{C}_{1}$ and $\mathcal{C}_{2}$ is determined by
\begin{equation}\label{2-4}
  u_{i}^{(1)}=u_{2i-1}\oplus u_{2i}, u_{i}^{(2)}=u_{2i}, i\in \mathcal{F}_{s}.
\end{equation} 

The WD of the polar code $\mathcal{C}$ is expressed as
\begin{equation}
  g_{\mathcal{C}}=\sum_{\hat{u}_{\mathcal{P}}\in \{0,1\}^{|\mathcal{P}|}}{g_{\mathcal{C}(\hat{u}_{\mathcal{P}})}}=\sum_{\hat{u}_{\mathcal{P}}\in \{0,1\}^{|\mathcal{P}|}}{g_{\mathcal{C}_{1}(\hat{u}_{\mathcal{P}})}\circledast g_{\mathcal{C}_{2}(\hat{u}_{\mathcal{P}})}}
\end{equation}
The WDs of component codes $\mathcal{C}_{1}(\hat{u}_{\mathcal{P}})$ and $\mathcal{C}_{2}(\hat{u}_{\mathcal{P}})$ can be computed in the same manner. This recursive process allows us enumerate the WDs of short codes and combine these results to derive the WDs of long codes. 

\begin{figure}[t]
     \centering
     \includegraphics[width=1\linewidth]{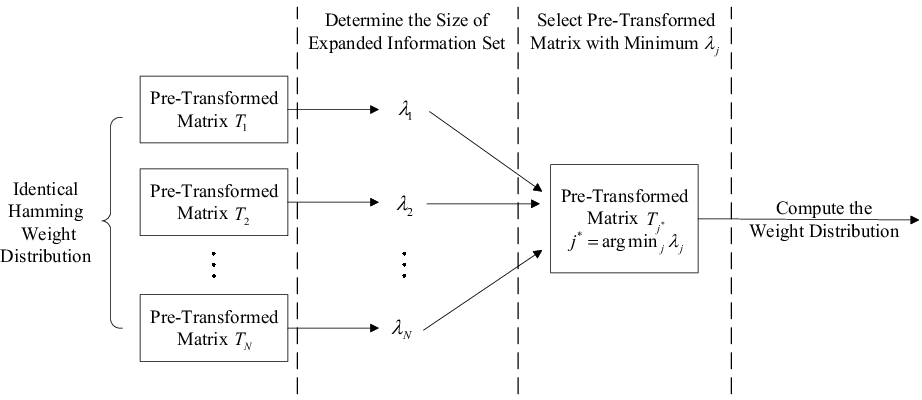}
     \caption{Block diagram of the proposed weight distribution computation of pre-transformed polar codes}
     \label{fig1}
\end{figure}

\section{WD Computation of Pre-Transformed Polar Codes Based on the PCD Method}
In this section, we first propose an iterative algorithm to identify the information bits that need to be expanded in pre-transformed polar codes. This expansion serves to eliminate the dependencies introduced by the pre-transformation and ensures compliance with the structure of a standard \emph{PCD coset}. Building upon this decomposition, we then establish a recursive method for calculating the Hamming weight distribution of pre-transformed polar codes.

\subsection{Determination of the Expanded Information Set}
In the sequence $v_{1}^{N}$, $v_{i}$ can be expressed as
\begin{equation}\label{2-6}
    v_{i} = u_{i} \oplus \sum_{j\leq i-1, T_{j,i}=1} u_{j}=u_{i} \oplus \sum_{j\leq i-1, T_{j,i}=1, j\in \mathcal{I}} u_{j}.
\end{equation}
When $u_{i}$ is a frozen bit, although $v_{i}$ and the preceding bits (indexed by $j\in \mathcal{I}$ with $T_{j,i}=1$) can take any value in $\{0,1\}$, they are not mutually independent. This dependency violates the prerequisite of the \emph{PCD method}. Consequently, it is necessary to expand the corresponding information bits over $\{0,1\}$ to resolve this coupling. We elaborate on this strategy in Proposition \ref{Proposition1}.

\begin{proposition}\label{Proposition1}
  In the pre-transformed polar codes $\mathcal{C}$, when $u_{i}$ is a frozen bit, then $u_{j}$ needs to be expanded as 0 and 1 in the \emph{PCD method} if $u_{j}$ is an information bit and $T_{j,i}=1$. Then $v_{i}$ and $v_{j}$ with $T_{j,i}=1$ are viewed as frozen bits in the PCD coset of $\mathcal{C}$.
\end{proposition}

\begin{proof}
Based on \eqref{2-6}, when $u_{j}$ is an information bit and $T_{j,i}=1$, it can take values from $\{0,1\}$. This makes the value of $v_{i}$ variable and dependent on $u_{j}$. In the \emph{PCD method}, if $v_{i}$ remains a variable dependent on an information bit from a different index, the standard recursive decomposition into independent sub-codes $\mathcal{C}_{1}$ and $\mathcal{C}_{2}$ cannot be directly applied, as the parity constraints would be ambiguous.

To resolve this coupling, the information bit $u_{j}$ must be \emph{expanded}. This means we decompose the original set of codewords into subsets (cosets) where $u_{j}$ takes a fixed value (either 0 or 1). When $u_{j}$ is fixed to a specific value $\alpha \in \{0,1\}$, the term $u_{j}$ in the summation for $v_{i}$ becomes a constant. Consequently, $v_{j}$ itself becomes determined (conditioned on previous bits) because its primary variable $u_{j}$ is now fixed. Similarly, $v_{i}$ becomes determined because the variable $u_{j}$ causing the uncertainty is fixed.

In the context of the \emph{PCD coset} defined by this specific expansion, the bits $v_{j}$ and $v_{i}$ behave as constants rather than free variables. Therefore, strictly within the scope of this PCD coset, $v_{j}$ and $v_{i}$ are effectively treated as frozen bits (carrying fixed values determined by the expansion and the matrix $T$). This allows the recursive calculation of WDs to proceed as if these positions were frozen in a standard polar code.
\end{proof}

After expanding the information bits in $u_{1}^{N}$ constrained by the frozen set, when $u_{i}\in \{0,1\}$, then $v_{i}\in {0,1}$ is an information bit. Consequently, we can conclude that the information and frozen sets of $v_{1}^{N}$ are identical to those of the expanded $u_{1}^{N}$.

\begin{algorithm}[!t]
    \caption{Determine the expanded information set}
    \label{alg:1}
    \renewcommand{\algorithmicrequire}{\textbf{Input:}}
    \renewcommand{\algorithmicensure}{\textbf{Output:}}
    \begin{algorithmic}[1]
        \REQUIRE Code length $N$, frozen set $\mathcal{F}$, matrix $T$  
        \ENSURE Expanded information set $\tilde{\mathcal{P}}$    
        \STATE $\tilde{\mathcal{P}}\leftarrow \varnothing$, $\mathcal{Q}\leftarrow \llbracket 1,N\rrbracket$
        \WHILE{$\mathcal{Q}\neq \varnothing$}
            \FOR {$i=N:1$}
                \IF{$i\in \mathcal{F}$ or $i\in \tilde{\mathcal{P}}$}
                    \STATE $\tilde{\mathcal{P}}\leftarrow \tilde{\mathcal{P}}\cup \{j|T_{j,i}=1,1\leq j\leq i-1\}$
                \ENDIF
            \ENDFOR
            \STATE $\mathcal{Q}^{\prime}\leftarrow \varnothing$
            \FOR {$i=1:N/2$}
                \IF{$2i-1\in \tilde{\mathcal{P}}\cup \mathcal{F}$ and $2i\notin \tilde{\mathcal{P}}\cup \mathcal{F}$}
                    \STATE $\mathcal{Q}^{\prime}\leftarrow \mathcal{Q}^{\prime}\cup 2i$
                \ELSIF{$2i-1\notin \tilde{\mathcal{P}}\cup \mathcal{F}$ and $2i\in \tilde{\mathcal{P}}\cup \mathcal{F}$}
                    \STATE $\mathcal{Q}^{\prime}\leftarrow \mathcal{Q}^{\prime}\cup 2i-1$
                \ENDIF
            \ENDFOR
            \STATE $\tilde{\mathcal{P}}\leftarrow \tilde{\mathcal{P}}\cup \mathcal{Q}^{\prime}$, $\mathcal{Q}\leftarrow \mathcal{Q}^{\prime}$
        \ENDWHILE
        \RETURN $\tilde{\mathcal{P}}$
    \end{algorithmic}
\end{algorithm}

\begin{remark} 
Based on Proposition 1, determining the value of $v_{i}$ for each $i\in \mathcal{F}$ requires expanding the information bit $u_{j}$ such that $u_{i}$ constrains $u_{j}$ where $j$ and $i$ satisfy $T_{j,i}=1$. However, after expanding $u_{j}$, the corresponding bit $v_{j}$ is no longer an information bit. Consequently, we must further expand the information bits $u_{k}$ constrained by $u_{j}$ to determine $v_{j}$, where $k$ and $j$ satisfy $T_{k,j}=1$. To identify the indices of information bits that need to be expanded, we sequentially scan from index $N$ down to index $1$ and search for frozen bits and the information bits constrained by them. This ordering avoids repeatedly revisiting previously processed information bits, since $T$ is an upper-triangular matrix.
\end{remark}

\begin{definition}
\label{def:couple_mapping}
The \emph{coupling mapping} $f$ is defined as
\begin{equation}\label{2-7}
  f(i) = \begin{cases}
    i-1, & \text{if } i \text{ is even}, \\
    i+1, & \text{if } i \text{ is odd}. 
  \end{cases}
\end{equation}
Furthermore, we refer to $v_{f(i)}$ as the \emph{couple bit} of $v_{i}$.
\end{definition}

Once expanding the information bits constrained by frozen bits, the remain information bits are mutually independent. Each coset can be viewed as an standard polar code and the \emph{PCD method} can be adopted to compute its WD. It is noted that there are additional information bits required to expand in the \emph{PCD method} when $v_{i}$ is an information bit but its couple bit $v_{f(i)}$ is a frozen bit.

\begin{definition}
In the pre-transformed polar codes, the \emph{expanded information set} $\tilde{\mathcal{P}}$ consisting of the information bits constrained by frozen bits or the information bits whose couple bit is a frozen bit. Furthermore, the \emph{full-frozen position set} and \emph{full-information position set} are defined as
\begin{align}
  \tilde{\mathcal{F}_{\text{s}}}&\triangleq \left\{m|2m-1\in \tilde{\mathcal{F}},2m\in \tilde{\mathcal{F}}\right\}, \\
  \tilde{\mathcal{I}_{\text{s}}}&\triangleq\left\{m|2m-1\in \tilde{\mathcal{I}},2m\in \tilde{\mathcal{I}}\right\},
\end{align} 
where $\tilde{\mathcal{F}}\triangleq \mathcal{F}\cup \tilde{\mathcal{P}}$ and $\tilde{\mathcal{I}}\triangleq \mathcal{I}\backslash \tilde{\mathcal{P}}$. 
\end{definition}

Algorithm \ref{alg:1} details the iterative procedure for constructing the expanded information set $\tilde{\mathcal{P}}$. The process is designed to satisfy two distinct types of constraints ensuring the applicability of the \emph{PCD method}. The procedure alternates between two phases until convergence. The first phase (Lines 3-7) enforces the dependency constraints derived from the frozen bits. By performing a backward scan based on uncoded sequence $u_{1}^{N}$, it identifies and expands the information bits that are constrained by the frozen set $\mathcal{F}$. As noted, once these constraints are resolved, the remaining information bits are mutually independent.
The second phase (Lines 9-15) addresses the structural requirements of the \emph{PCD method}. Specifically, when $v_{i}$ is an information bit but its couple bit $v_{f(i)}$ is a frozen bit, the \emph{PCD method} requires additional expansion. The algorithm detects such mixed pairs and forces the unconstrained bit into $\tilde{\mathcal{P}}$. This expansion is repeated iteratively via the set $\mathcal{Q}$ until all constraints are satisfied.
\subsection{WD Computation for Pre-Transformed Polar Codes}
\begin{theorem}
Given a pre-transformed polar code $\mathcal{C}$ with the generator matrix ${G}_{N}={K}_{2}^{\otimes n}$ and pre-transformed matrix $T$, its WD ${g}_{\mathcal{C}}$ can be calculated as
\begin{equation}\label{eq13}
  {g}_{\mathcal{C}}=\sum\limits_{\hat{u}_{\tilde{\mathcal{P}}}\in {{\{0,1\}}^{|\tilde{\mathcal{P}}|}}}{{g}_{{{\mathcal{C}}}(\hat{u}_{\tilde{\mathcal{P}}})}}=\sum\limits_{\hat{u}_{\tilde{\mathcal{P}}}\in {{\{0,1\}}^{|\tilde{\mathcal{P}}|}}}{{g}_{{{\mathcal{C}}_{1}}(\hat{u}_{\tilde{\mathcal{P}}})}\circledast {g}_{{{\mathcal{C}}_{2}}(\hat{u}_{\tilde{\mathcal{P}}})}}
\end{equation}
where $\mathcal{C}_{i}(\hat{u}_{\tilde{\mathcal{P}}})$, for $i=1,2$, is a standard polar code with the generator matrix ${K}_{2}^{\otimes n-1}$, information set $\tilde{\mathcal{I}}_{s}$ and frozen set $\tilde{\mathcal{F}}_{s}$. For the code $\mathcal{C}_{i}(\hat{u}_{\tilde{\mathcal{P}}})$, the frozen bit in position $m \in \tilde{\mathcal{F}}_{s}$ is given by the $i$-th bit of $(v_{2m-1},v_{2m}){K}_{2}$ after expanding the information bits $\hat{u}_{\tilde{\mathcal{P}}}$ as $\{0,1\}^{|\tilde{\mathcal{P}}|}$. 
\end{theorem}
\begin{proof}
The calculation of $g_{\mathcal{C}}$ relies on decomposing the pre-transformed polar code into disjoint cosets by expanding the information bits in $\tilde{\mathcal{P}}$.
According to Algorithm \ref{alg:1} and Proposition 1, the construction of $\tilde{\mathcal{P}}$ ensures two key properties for the uncoded sequence $v_{1}^{N}=u_{1}^{N}T$:
\begin{enumerate}
    \item \emph{Decoupling:} The dependencies introduced by the pre-transformation matrix $T$ are resolved. Specifically, after fixing the bits in $\tilde{\mathcal{P}}$, the remaining information bits in the set $\mathcal{I} \setminus \tilde{\mathcal{P}}$ are mutually independent.
    \item \emph{Structural Consistency:} For any pair $(v_{2m-1}, v_{2m})$, the bits are either both effectively frozen (belonging to $\mathcal{F} \cup \tilde{\mathcal{P}}$) or both information bits (belonging to $\mathcal{I} \setminus \tilde{\mathcal{P}}$).
\end{enumerate}

Consequently, each specific coset $\mathcal{C}(\hat{u}_{\tilde{\mathcal{P}}})$ characterized by $\hat{u}_{\tilde{\mathcal{P}}} \in \{0,1\}^{|\tilde{\mathcal{P}}|}$ possesses an effective information set $\mathcal{I} \setminus \tilde{\mathcal{P}}$ and an effective frozen set $\mathcal{F} \cup \tilde{\mathcal{P}}$. This structure aligns with the definition of a standard \emph{PCD coset} in \eqref{2-1}. Therefore, the \emph{PCD method} in \cite{Liu2506:Parity} is directly applicable to the coset $\mathcal{C}(\hat{u}_{\tilde{\mathcal{P}}})$. The WD of each coset $\mathcal{C}(\hat{u}_{\tilde{\mathcal{P}}})$ can be recursively calculated as the convolution of its two component codes: $\mathcal{C}_1(\hat{u}_{\tilde{\mathcal{P}}})$ and $\mathcal{C}_2(\hat{u}_{\tilde{\mathcal{P}}})$.
These two component codes correspond to the $N/2$-dimension polar codes with information set $\tilde{\mathcal{I}}_{s}$ and frozen set $\tilde{\mathcal{F}}_{s}$ which substitute the information set and frozen set in \eqref{2-2} and \eqref{2-3} by $\mathcal{I}\backslash \tilde{\mathcal{P}}$ and $\mathcal{F}\cup \tilde{\mathcal{P}}$, respectively. The values of the frozen bits for these component codes are explicitly determined by \eqref{2-4}, replacing $u_{1}^{N}$ and $\mathcal{F}_{s}$ by $v_{1}^{N}$ and $\tilde{\mathcal{F}}_{s}$, respectively. Summing over all valid $\hat{u}_{\tilde{\mathcal{P}}}$ yields the total WD as stated in \eqref{eq13}.
\end{proof}

Based on Theorem 1, we can characterize a pre-transformed polar code as a union of \emph{PCD cosets}. Crucially, each coset is formed by the concatenation of two component codes, both of which are standard polar codes. This allows us to apply the \emph{PCD method} \cite{Liu2506:Parity} recursively: we decompose the component codes layer by layer until the code length falls below a specific threshold. Finally, the total weight distribution is derived by convolving the distributions of the leaf component codes and summing them over the set of all cosets.

\subsection{Numerical Results}
Table \ref{table:1} compares the sizes of the expanded information sets obtained by the method in \cite{yao2024deterministic} and the proposed \emph{PCD method} for pre-transformed polar codes with a length of $N=128$ across various code rates. Here, $n_{1}$ and $n_{2}$ denote the sizes of the expanded information set in \cite{yao2024deterministic} and \emph{PCD method}, respectively. Specifically, we utilize PC-Polar codes where the parity-check equation for each parity bit is defined as $v_{i}=u_{i}+u_{i-3}+u_{i-5}+u_{i-6}$. In \cite{yao2024deterministic}, the recursive calculation requires expanding all information bits preceding the frozen bit with the largest index into 0 and 1. Consequently, the size of its expanded information set is equal to the number of information bits located before this last frozen bit. The information set is constructed based on the 5G NR Polar sequence\cite{3gpp38212}. Results indicate that as the code rate increases, the \emph{PCD method} achieves a significant reduction in the expanded information set size compared to \cite{yao2024deterministic}. Furthermore, the size of the expanded information set remains relatively small for the \emph{PCD method} at both high and low code rates.

However, we observe that for general random pre-transformed polar codes and most PAC codes, the \emph{PCD method} achieves negligible reduction in the expanded information set size compared to the method in \cite{yao2024deterministic}. To address this limitation, in the next section, we introduce equivalence classes of pre-transformed polar codes sharing identical Hamming weight distributions to further reduce the set size.

\begin{table}
\renewcommand{\arraystretch}{1.3}
\centering
\caption{Expanded Information Set Sizes for PC-Polar Codes under Different Code Rates (N=128)}
\label{table:1}
\begin{tabular}{|c|c|c|c|c|c|c|c|c|}
\hline
\multirow{1}{*}{$K$} & $n_{1}$ & $n_{2}$ & $K$ & $n_{1}$ & $n_{2}$ & $K$ & $n_{1}$ & $n_{2}$ \\
\hline
12 & 6 & 4 & 51 & 29 & 23 & 89 & 43 & 27 \\ 
\hline
25 & 13 & 13 & 64 & 34 & 28 & 102 & 42 & 22 \\ 
\hline
38 & 24 & 20 & 76 & 46 & 30 & 115 & 53 & 17 \\ 
\hline
\end{tabular}
\end{table}

\section{Equivalence Classes of Pre-Transformed Polar Codes for the PCD Method}
In this section, we first prove that pre-transformation matrices derived from the row vectors of standard polar codes induce a cyclic shift on the codewords. Based on this property, we establish an equivalence class for any pre-transformed polar code that maintains the same Hamming weight distribution. We then minimize the computational complexity by selecting the code with the smallest expanded information set size from this class for the PCD calculation.
\subsection{Equivalence Classes of Pre-Transformed Polar Codes}
We first define the cyclic right shift operation for vectors and then present a lemma that serves as the basis for the equivalence classes of pre-transformed polar codes sharing identical Hamming weight distributions.

\begin{definition}
\label{def:cyclic_shift}
Let $u = (u_{1}, u_{2}, \dots, u_{N}) \in \{0,1\}^{N}$ be a binary vector. The \emph{cyclic right-shift operator} by $s$ positions, denoted as $\sigma_{s}$, maps $u$ to a new vector $\sigma_{s}(u)$ given by:
\begin{equation}
    \sigma_{s}(u) = (u_{N-s+1}, \dots, u_{N}, u_{1}, \dots, u_{N-s}).
\end{equation}
The cyclic shift operation is \emph{linear} over the binary field $\mathbb{F}_2$.
\end{definition}

\begin{lemma}\label{Lemma1}
For any arbitrary row vectors $g_{i}$ and $g_{j}$ of the polar code generator matrix $K_{2}^{\otimes n}$, the summation $\sum_{k:g_{j}(k)=1,k\le N+1-i}{g_{i+k-1}}$ results in $\sigma_{j-1}(g_{i})$.
\end{lemma}

The proof of Lemma \ref{Lemma1} is given by Appendix \ref{AppendixA}. Lemma \ref{Lemma1} specifies which row vectors need to be linearly combined to obtain a cyclic shift of the vector $g_{i}$. Specifically, the selection of these vectors is determined by the positions of the ones in $g_{j}$, with a shift amount of $j-1$. To facilitate the understanding of Lemma 1, we provide a simple example.

\begin{example}
  In the generator matrix $K_{2}^{\otimes 4}$, we set $i=7$ and $j=6$.  The corresponding row vectors are given by
  \begin{align}
    g_{6}&=(1,1,0,0,1,1,0,0,0,0,0,0,0,0,0,0), \nonumber \\
    g_{7}&=(1,0,1,0,1,0,1,0,0,0,0,0,0,0,0,0). \nonumber
  \end{align}
  The set of indices where the vector $g_{6}$ takes the value 1 is given by $\{1,2,5,6\}$. The result $\sum_{k:g_{6}(k)=1}{g_{k+6}}$ is given by
  \begin{equation*}
    g_{7}+g_{8}+g_{11}+g_{12}=(0,0,0,0,0,1,0,1,0,1,0,1,0,0,0,0).
  \end{equation*}
  This result corresponds exactly to the cyclic right shift of vector $g_{7}$ by 5 positions.
\end{example}

\begin{definition}
Let $g_{i}$ be a row vector of the generator matrix. The pre-transformation matrix generated by $g_{i}$, denoted as $T(g_{i})$, is an $N \times N$ matrix where the $j$-th row is obtained by logically right-shifting $g_{i}$ by $j-1$ positions and truncating the overflow bits. Specifically, the $j$-th row is given by
\begin{equation}
    T(g_{i})_{j} = (0_{1}^{j-1}, g_{i}(1), \dots, g_{i}(N-j+1)),
\end{equation}
where $g_{i}(k)$ denotes the $k$-th element of $g_{i}$.
\end{definition}

In the following, we characterize the equivalence class for an arbitrary pre-transformed polar code, induced by the identification of the Hamming weight distribution.

\begin{theorem}\label{Theorem2}
For a fixed information set, the Hamming weight distribution of a pre-transformed polar code based on the pre-transformation matrix $T^{(0)}$ is identical to that of the code based on the matrix $T^{(0)}T(g_{i})$.
\end{theorem}
\begin{proof}
Let $\tilde{G} = T(g_{j})K_{2}^{\otimes n}$. The $i$-th row of $\tilde{G}$, denoted by $\tilde{g}_{i}$, is given by $\tilde{g}_{i}= \sum_{k:g_{j}(k)=1,k\le N+1-i}{g_{i+k-1}}$. According to Lemma \ref{Lemma1}, this row corresponds to the cyclic shift of the original basis vector, i.e., $\tilde{g}_{i}=\sigma_{j-1}(g_{i})$.

Consequently, let $G'$ denote the generator matrix of the new code, i.e., $G' = T^{(0)} \tilde{G} = T^{(0)} T(g_{j})K_{2}^{\otimes n}$. The $i$-th row of $G'$, denoted by $g'_{i}$, can be derived as:
\begin{align}
    g'_{i} &= \sum_{k: T_{i}^{(0)}(k)=1}{\sigma_{j-1}(g_{k})} = \sigma_{j-1}\left(\sum_{k: T_{i}^{(0)}(k)=1}{g_{k}}\right). \label{eq:linearity}
\end{align}
Equation \eqref{eq:linearity} holds due to the linearity of the cyclic shift operation. Note that the term $\sum_{k: T_{i}^{(0)}(k)=1}{g_{k}}$, represents precisely the $i$-th row of the original generator matrix $T^{(0)} K_{2}^{\otimes n}$.

Thus, the $i$-th row of the new generator matrix $T^{(0)} T(g_{j})K_{2}^{\otimes n}$ is simply the cyclic shift (by $j-1$ positions) of the $i$-th row of the original generator matrix $T^{(0)} K_{2}^{\otimes n}$.
Given that the information set remains unchanged, the mapping $\sigma_{j-1}$ establishes a one-to-one correspondence (isomorphism) between the codewords generated by $T^{(0)} K_{2}^{\otimes n}$ and those generated by $T^{(0)} T(g_{j})K_{2}^{\otimes n}$. Since the cyclic shift $\sigma_{j-1}$ preserves the Hamming weight of any vector, the Hamming weight distributions of these two codes are identical.
\end{proof}

\begin{remark}
  Based on Theorem \ref{Theorem2}, for any arbitrary pre-transformed polar code, a corresponding equivalence class can be identified. The cardinality of each equivalence class is equal to the code length $N$, and all codes within the class share the same information and frozen sets. By transforming the original pre-transformation matrix $T^{(0)}$ into $T^{(0)}T(g_{i})$, it is possible to sparsify the matrix (i.e., reduce the density of ones). This increased sparsity reduces the number of information bits constrained by frozen bits, thereby decreasing the size of the expanded information set. For instance, consider a PAC code with a memory vector of $(1,1,1,1,1)$. Applying the transformation $T(g_{2})$ results in a new memory vector $(1,0,0,0,0,1)$ while maintaining an invariant Hamming weight distribution. Consequently, the equivalent pre-transformation matrix becomes significantly sparser.
\end{remark}

\subsection{Optimized PCD Method via Equivalence Classes}
According to Theorem \ref{Theorem2}, the Hamming weight distribution of a polar code defined by the pre-transformation matrix $T$ is identical to that defined by $T T(g_{j})$. This implies that for any pre-transformed polar code, there exist $N-1$ other equivalent codes. For each candidate matrix $T T(g_{j})$, we can derive the corresponding expanded information set $\tilde{\mathcal{P}}_{j}$ using Algorithm \ref{alg:1}. Consequently, we select $T T(g_{j^{*}})$ as the optimal pre-transformation matrix, where $j^{*}=\operatorname{argmin}_{j}|\tilde{\mathcal{P}}_{j}|$.

Table \ref{table:2} presents the proportion of random pre-transformed polar codes that achieve a reduction in the size of the expanded information set compared to the method in \cite{yao2024deterministic}. Specifically, $r_{1}$ denotes the proportion of codes exhibiting such a reduction using the \emph{PCD method} without equivalence class optimization, while $r_{2}$ represents the proportion achieved when equivalence class optimization is applied. These statistics are derived from Monte Carlo simulations for code lengths of $N=128$ and $N=256$ across various code rates, utilizing the 5G NR Polar sequence\cite{3gpp38212} for code construction.

In the absence of equivalence class optimization, the \emph{PCD method} reduces the number of expanded information bits for approximately 6.7\% of the random pre-transformed polar codes across most configurations compared to \cite{yao2024deterministic}. By contrast, when equivalence class optimization is employed, this proportion significantly increases to approximately 20\%.

Table \ref{table:3} compares the expanded information set sizes before and after optimization for PAC codes with length $N=128$ and a memory vector of $(1,0,1,0,1,0,1,1)$. Let $n_{1}$, $n_{2}$ and $n_{3}$ denote the sizes of the expanded information sets obtained by the method in \cite{yao2024deterministic}, the \emph{PCD method} without equivalence class optimization, and the \emph{PCD method} with equivalence class optimization, respectively. Under every code rate, the \emph{PCD method} without equivalence class optimization yields the same expanded information set size as the method in \cite{yao2024deterministic}. After equivalence class optimization, the memory vector is given by $(1,0,0,0,0,0,0,1,1,1)$. The results demonstrate a reduction in size across all code rates. Notably, a maximum reduction of 4 bits is achieved, which reduces the computational complexity to $1/16$ of the original level.

\begin{table}
\renewcommand{\arraystretch}{1.3}
\centering
\caption{Average Reduction Ratio of the Expanded Information Set Size for Random Pre-transformed Polar Codes}
\label{table:2}
\begin{tabular}{|c|c|c|c|c|c|}
\hline
\multicolumn{3}{|c|}{$N = 128$} & \multicolumn{3}{c|}{$N = 256$} \\ 
\cline{1-6}
\multirow{1}{*}{$K$} & $r_{1}$ & $r_{2}$ & $K$ & $r_{1}$ & $r_{2}$ \\ 
\hline
25 & 0.4\% & 1.52\% & 51 & 6.66\% & 21.21\% \\ 
\hline
38 & 6.63\% & 20.51\% & 76 & 6.65\% & 20.42\% \\ 
\hline
51 & 6.28\% & 18.73\% & 102 & 6.72\% & 23.43\% \\ 
\hline
64 & 6.61\% & 21.21\% & 128 & 6.29\% & 18.91\% \\ 
\hline
76 & 6.56\% & 22.24\% & 153 & 6.70\% & 22.44\% \\ 
\hline
89 & 6.65\% & 20.63\% & 179 & 6.57\% & 19.72\% \\ 
\hline
102 & 0.4\% & 2.92\% & 204 & 6.70\% & 21.02\% \\ 
\hline
115 & 6.65\% & 21.61\% & 230 & 6.64\% & 22.24\% \\ 
\hline
\end{tabular}
\end{table}

\begin{table}
\renewcommand{\arraystretch}{1.3}
\centering
\caption{Reduction of the Expanded Information Set Size for a 128-length PAC code with memory $(1,0,1,0,1,0,1,1)$}
\label{table:3}
\begin{tabular}{|c|c|c|c|c|c|c|c|}
\hline
\multirow{1}{*}{$K$} & $n_{1}$ & $n_{2}$ & $n_{3}$ & $K$ & $n_{1}$ & $n_{2}$ & $n_{3}$ \\
\hline
25 & 13 & 13 & 11 & 76 & 46 & 46 & 42 \\ 
\hline
38 & 24 & 24 & 20 & 89 & 43 & 43 & 39 \\ 
\hline
51 & 29 & 29 & 27 & 102 & 42 & 42 & 40 \\ 
\hline
64 & 34 & 34 & 30 & 115 & 53 & 53 & 49 \\
\hline
\end{tabular}
\end{table}

\section{Conclusion}
This paper proposes an efficient recursive algorithm based on the Parity-Consistent Decomposition (PCD) method to determine the complete weight distribution of pre-transformed polar codes. We introduce an \emph{Expanded Information Set} to resolve the bit dependencies caused by pre-transformation, thereby enabling the recursive decomposition. Furthermore, by exploiting code equivalence classes to minimize the size of this set, we achieve a substantial reduction in computational overhead. Numerical results confirm that the proposed method significantly outperforms existing deterministic algorithms in terms of complexity.

\appendix
\subsection{Proof of Lemma \ref{Lemma1}}\label{AppendixA}
\begin{proposition}\label{Lemma2}
Let $e_{i}=(0_{1}^{i-1},1,0_{i+1}^{N})$ denote the standard basis vector with a 1 at the $i$-th position. It can be expressed as $e_{i}=\sum_{j:g_{i}(j)=1}{g_{j}}$, where $g_{l}$ denotes the $l$-th row vector of the polar code generator matrix $K_{2}^{\otimes n}$ for $l\in \llbracket 1,N\rrbracket$.
\end{proposition}

\begin{proof}
According to \cite{5075875}, the generator matrix of polar codes satisfies the self-inverse property $K_{2}^{\otimes n}K_{2}^{\otimes n}=I_{N}$. For the identity matrix $I_{N}$, the $i$-th row vector is $e_{i}$. For the product $K_{2}^{\otimes n}K_{2}^{\otimes n}$, the $i$-th row vector is given by $g_{i}K_{2}^{\otimes n}$. Since the matrix multiplication corresponds to a linear combination of rows, we have $g_{i}K_{2}^{\otimes n}=\sum_{j:g_{i}(j)=1}{g_{j}}$.
\end{proof}

\begin{proposition}\label{Lemma3}
The row vectors of $K_{2}^{\otimes n}$ satisfy the property: $g_{k}=g_{i} \circledast_{N} g_{j}$, where $i$ and $j$ satisfy $i+j-1=k$, and $\circledast_{N}$ denotes the operation of taking the first $N$ bits of the convolution result.
\end{proposition}

\begin{proof}
We proceed by mathematical induction on $n$. For any arbitrary $k$, it suffices to verify that the relation $g_{k}=g_{i}\circledast_{N}g_{k+1-i}$ holds for all $1\le i\le k$.

\emph{Base Case ($n=1$):}
For $N=2$, the row vectors are $g_1=(1,0)$ and $g_2=(1,1)$. It is straightforward to verify that $g_{1}=g_{1} \circledast_{2} g_{1}$ and $g_{2}=g_{1} \circledast_{2} g_{2}=g_{2} \circledast_{2} g_{1}$. The property holds.

\emph{Inductive Step:}
Assume that for $K_{2}^{\otimes n-1}$, the row vectors satisfy $g_{k}=g_{i} \circledast_{N/2} g_{j}$ for $i+j-1=k$. Now consider the matrix $K_{2}^{\otimes n}$ with block length $N$.

\emph{Case 1:} $1\le k\le 2^{n-1}$. Due to the recursive structure of $K_{2}^{\otimes n}$, the first $N/2$ rows effectively correspond to the rows of $K_{2}^{\otimes n-1}$ padded with zeros. Thus, we directly obtain $g_{k}=g_{i} \circledast_{N} g_{j}$ for $i+j-1=k$.

\emph{Case 2:} $k=2^{n-1}+1$. First, observe that $g_{2^{n-1}+1}=g_{1} \circledast_{N} g_{2^{n-1}+1}=g_{2^{n-1}+1} \circledast_{N} g_{1}$.
Additionally, since $g_{2^{n-1}+1}=(1,0_{2}^{N/2},1,0_{N/2+2}^{N})$ and $g_{2^{n-1}}=(1_{1}^{N/2},0_{N/2+1}^{N})$, calculation shows that $g_{2^{n-1}+1}=g_{2} \circledast_{N} g_{2^{n-1}}=g_{2^{n-1}} \circledast_{N} g_{2}$.
From the inductive hypothesis, we know $g_{2^{n-1}}=g_{i} \circledast_{N} g_{j}$ where $i+j=2^{n-1}+1$.
Substituting this into the previous equation, $g_{2^{n-1}+1}$ can be expressed as:
\begin{equation}
    g_{2^{n-1}+1}=g_{i} \circledast_{N} g_{j} \circledast_{N} g_{2}.
\end{equation}
For $3\le i\le 2^{n-1}$, we have $j+2\le 2^{n-1}$. Applying the inductive hypothesis again, we have $g_{j} \circledast_{N} g_{2}=g_{j+1}$.
Therefore, $g_{2^{n-1}+1}=g_{i} \circledast_{N} g_{j+1}$.

In summary, for $k=2^{n-1}+1$, the property $g_{k}=g_{i} \circledast_{N} g_{j}$ holds where $i+j-1=k$.

\emph{Case 3:} $2^{n-1}+2\le k\le 2^{n}$. Recall the recursive definition $K_{2}^{\otimes n}=\begin{bmatrix} K_{2}^{\otimes n-1} & 0 \\ K_{2}^{\otimes n-1} & K_{2}^{\otimes n-1} \end{bmatrix}$.
Let $g_{k'}$ denote a row in $K_{2}^{\otimes n-1}$ satisfying $g_{k'}=g_{i'} \circledast_{N/2} g_{j'}$ with $i'+j'-1=k'$.
The row $g_{k'+N/2}$ in $K_{2}^{\otimes n}$ is constructed as $(g_{k'},g_{k'})$. Based on the convolution property, this can be written as:
\begin{equation}\label{4-12}
  g_{k'+N/2}=(g_{i'} \circledast_{N/2} g_{j'}, g_{i'} \circledast_{N/2} g_{j'}) = g_{i'} \circledast_{N} g_{j'+N/2}
\end{equation}
Now, for any $i, j$ satisfying $i+j-1=k$, we discuss three sub-cases for $i$:
\begin{enumerate}
    \item If $1\le i\le k-2^{n-1}$, then $j\ge 2^{n-1}+1$. The relation $g_{k}=g_{i} \circledast_{N} g_{j}$ holds directly from \eqref{4-12}.
    \item If $k-2^{n-1}<i\le 2^{n-1}+1$, then $k-2^{n-1}<j<2^{n-1}+1$. We have:
    \begin{align}\label{Th21}
       g_{i} \circledast_{N} g_{j} &\overset{(a)}{=} g_{i} \circledast_{N} g_{2^{n-1}+2-i} \circledast_{N} g_{i+j-2^{n-1}-1} \nonumber \\
       &= g_{2^{n-1}+1} \circledast_{N} g_{i+j-2^{n-1}-1} \overset{(b)}{=} g_{k}
    \end{align}
    where (a) follows from Case 1 and (b) follows from equation \eqref{4-12}.
    \item If $1+2^{n-1}\le i\le k$, then by equation \eqref{4-12}, we have $g_{i} \circledast_{N} g_{j}=g_{k}$. 
\end{enumerate}
\end{proof}

\emph{Proof of Lemma \ref{Lemma1}}: 
\emph{Case 1:} $1\le j\le N+1-i$. For the row vector $g_{i}$, the last $N-i$ elements are all zeros. In this scenario, $\sigma_{j-1}(g_{i})$ is equivalent to
\begin{align}
    g_{i} \circledast_{N} e_{j} &\overset{(a)}{=} g_{i} \circledast_{N} \left(\sum_{k:g_{j}(k)=1}{g_{k}}\right) = \sum_{k:g_{j}(k)=1}{g_{i} \circledast_{N} g_{k}} \nonumber \\
    &\overset{(b)}{=} \sum_{k:g_{j}(k)=1}{g_{i+k-1}}
\end{align}
where (a) and (b) follow from Propositins \ref{Lemma2} and \ref{Lemma3}, respectively.

\emph{Case 2:} $N+1-i<j\le N$. Consider the extended matrix $K_{2}^{\otimes n+1}=\begin{bmatrix} K_{2}^{\otimes n} & 0 \\ K_{2}^{\otimes n} & K_{2}^{\otimes n} \end{bmatrix}$. In this matrix, the row vector $g_{i}$ is followed by $2N-i$ zeros. Let $g_{i}^{(N)}$ denote the $i$-th row in $K_{2}^{\otimes n}$ and $g_{i}^{(2N)}$ denote the $i$-th row in $K_{2}^{\otimes n+1}$.
The convolution of $g_{i}^{(2N)}$ with $e_{j}$ in the $2N$-dimensional space corresponds to a right shift by $j-1$ positions:
\begin{align}
    g_{i}^{(2N)} \circledast_{2N} e_{j} &= \sum_{k:g_{j}^{(N)}(k)=1}{g_{i+k-1}^{(2N)}} \\
    &= \sum_{\substack{k:g_{j}^{(N)}(k)=1 \\ k\le N+1-i}}{g_{i+k-1}^{(2N)}} + \sum_{\substack{k:g_{j}^{(N)}(k)=1 \\ k>N+1-i}}{g_{i+k-1}^{(2N)}} \nonumber
\end{align}
For the first term, since $i+k-1\le N$, it can be rewritten as:
\begin{equation}
  \sum_{\substack{k:g_{j}^{(N)}(k)=1,\\ k\le N+1-i}}{g_{i+k-1}^{(2N)}} = \sum_{\substack{k:g_{j}^{(N)}(k)=1,\\ k\le N+1-i}}{(g_{i+k-1}^{(N)},0_{N+1}^{2N})}
\end{equation}
Let us denote the result of this first term as $(a_{1}^{N},0_{1}^{N})$.

For the second term, we have:
\begin{equation}
  \sum_{\substack{k:g_{j}^{(N)}(k)=1,\\ k>N+1-i}}{g_{i+k-1}^{(2N)}} = \sum_{\substack{k:g_{j}^{(N)}(k)=1,\\ k>N+1-i}}{(g_{i+k-1-N}^{(N)},g_{i+k-1-N}^{(N)})}
\end{equation}
Since $k\le j$, the ones in vector $g_{i+k-1-N}^{(N)}$ can only exist in the first $i+j-1-N$ positions. Given that $i+j-1-N\le j-1$, we can express the result of the second term as $(w_{1}^{j-1},0_{j}^{N},w_{1}^{j-1},0_{N+j}^{2N})$, where $w_{1}^{j-1}$ corresponds to the first $j-1$ elements of the summation $\sum_{k>N+1-i}{g_{i+k-1-N}^{(N)}}$.

Thus, the full convolution is:
\begin{equation}\label{eq21}
  g_{i}^{(2N)} \circledast_{2N} e_{j} = (a_{1}^{j-1}+w_{1}^{j-1}, a_{j}^{N}, w_{1}^{j-1}, 0_{N+j}^{2N})
\end{equation}
Since $i\le N$, we have $g_{i}^{(N)}=g_{i}^{(2N)}$ (in the first $N$ bits), and thus the simple shift implies:
\begin{equation}\label{eq22}
  g_{i}^{(2N)} \circledast_{2N} e_{j} = (0_{1}^{j-1}, g_{i}^{(N)}, 0_{j+N}^{2N})
\end{equation}
Comparing \eqref{eq21} and \eqref{eq22}, we identify that $w_{1}^{j-1}=(g_{i}^{(N)}(N-j+2),\cdots,g_{i}^{(N)}(N))$, which corresponds to the elements from $N+1$ to $N+j-1$ of the shifted vector. Furthermore, $a_{1}^{j-1}+w_{1}^{j-1}=0_{1}^{j-1}$ implies $a_{1}^{j-1}=w_{1}^{j-1}$. Additionally, it is derived that $a_{j}^{N}=(g_{i}^{(N)}(1),\cdots, g_{i}^{(N)}(N-j+1))$. The vector $a_{1}^{N}$ can be expressed as
\begin{equation*}
  (g_{i}^{(N)}(N-j+2),\cdots,g_{i}^{(N)}(N),g_{i}^{(N)}(1),\cdots, g_{i}^{(N)}(N-j+1))
\end{equation*}

Consequently, the term $\sum_{k:g_{j}^{(N)}(k)=1,k\le N+1-i}{g_{i+k-1}^{(N)}}=a_{1}^{N}$ effectively takes the part of the shifted vector that overflowed beyond length $N$ (the $w_{1}^{j-1}$ part) and moves it to the beginning. This precisely describes a cyclic right shift of $g_{i}^{(N)}$ by $j-1$ positions. $\hfill\square$




\bibliographystyle{IEEEtran} 
\bibliography{citation}    

\end{document}